\documentclass[11pt]{article}
\usepackage[cp1251]{inputenc}
\usepackage[T2A]{fontenc}
\usepackage[english]{babel}
\usepackage{amsthm, amsmath, amssymb, amsbsy, amscd, amsfonts, latexsym, euscript,epsf}
\newtheorem{thm}{Theorem}[section]
\usepackage{bbold, bm}
\usepackage{epic,eepic}
\usepackage{texdraw}
\usepackage[dvips]{color}
\usepackage{color}
\usepackage{graphicx}
\usepackage{exscale,relsize}
\usepackage{authblk}
\usepackage{url}
\usepackage{enumerate}
\allowdisplaybreaks
\newtheorem{proposition}[thm]{Proposition}
\newtheorem{corollary}[thm]{Corollary}

\newtheorem{remark}[thm]{Remark}
\newtheorem{theorem}[thm]{Theorem}

\title{\textbf{Nonlinear Schr\"odinger type tetrahedron maps}} 

\author{S. Konstantinou-Rizos\thanks{s.konstantinu.rizos@uniyar.ac.ru, skonstantin84@gmail.com}}

\affil{Centre of Integrable Systems, P.G. Demidov Yaroslavl State University, Yaroslavl, Russia}

\theoremstyle{definition}

\DeclareMathOperator{\tr}{tr}

\DeclareMathOperator{\End}{End}

\DeclareMathOperator{\id}{id}

\begin{document}

\maketitle

\begin{abstract}
This paper is concerned with the construction of new solutions in terms of birational maps to the functional tetrahedron equation and parametric tetrahedron equation. We present a method for constructing solutions to the parametric tetrahedron equation via Darboux transformations. In particular, we study matrix refactorisation problems for Darboux transformations associated with the nonlinear Schr\"odinger (NLS) and the derivative nonlinear Schr\"odinger (DNLS) equation, and we construct novel nine-dimensional tetrahedron maps. We show that the latter can be restricted to six-dimensional parametric tetrahedron maps on invariant leaves. Finally, we construct parametric tetrahedron maps employing degenerated Darboux transformations of NLS and DNLS type. 
\end{abstract}

\bigskip

\hspace{.2cm} \textbf{PACS numbers:} 02.30.Ik, 02.90.+p, 03.65.Fd.

\hspace{.2cm} \textbf{Mathematics Subject Classification 2020:} 35Q55, 16T25.

\hspace{.2cm} \textbf{Keywords:} Tetrahedron equation, Parametric tetrahedron maps, Darboux transformations,

\hspace{2.4cm} NLS type equations, Yang--Baxter maps.

\section{Introduction}
The functional tetrahedron equation is a higher-dimensional generalisation of the Yang--Baxter equation, one of the most fundamental equations of Mathematical Physics, and it was first studied by Zamolodchikov in \cite{Zamolodchikov, Zamolodchikov-2}. Ever since, it has attracted the interest of many scientists in the area of Mathematical Physics (indicatively, we refer to \cite{Bazhanov-Sergeev, Bazhanov-Sergeev-2, Dimakis, Doliwa-Kashaev, Gorbounov-Talalaev, Kashaev, Kashaev-Sergeev, Kassotakis-Tetrahedron, Sharygin-Talalaev, Nijhoff, Sergeev}) who studied its solutions from various aspects. 

From the point of view of classical integrable systems, an important result is Sergeev's classification of tetrahedron maps in \cite{Sergeev} where he studied the relation between matrix trifactorisation problems and the tetrahedron equation. Another important result is the connection between the tetrahedron equation and integrable systems on the three-dimensional lattice. In particular, this connection was established in \cite{Kassotakis-Tetrahedron}, as a generalisation of the ideas presented in \cite{Pap-Tongas-Veselov}, and tetrahedron maps together with their vector generalisations were constructed using the invariants of symmetry groups of three-dimensional lattice systems.  Moreover, noncommutative versions of integrable systems have been of great interest over the past few decades, due to the plethora of applications in Physics. In fact, fully noncommutative versions of solutions to the tetrahedron equation have already been found (see, for example, \cite{Doliwa-Kashaev}). Of course, the results related to the tetrahedron equation are not limited to the aforementioned, and there are plenty of other results that highlight the importance of the functional tetrahedron equation. However, the functional tetrahedron equation has not yet reached the same level of attention as its lower-dimensional analogue, the Yang--Baxter equation. This is probably due to the fact that it is more difficult to find solutions to it, and not many methods are yet available for constructing such solutions.

In this paper, we present a method for constructing solutions to the parametric functional tetrahedron equation based on the ideas presented in \cite{Sokor-Sasha}. In particular, we study matrix refactorisation problems for certain Darboux matrices and we derive solutions to the classical (without parameters) functional tetrahedron equation. The entries of each Darboux matrix satisfy a system of differential equations, the so-called B\"acklund transformation. The latter system admits a first integral which indicates the existence of invariant leaves. On these invariant leaves our derived solutions are expressed in terms of parametric birational maps which satisfy the parametric functional tetrahedron equation.

\subsection{Organisation of the paper}
The paper is organised as follows. 

In the next section, we present all the necessary definitions for the text to be self-contained. In particular, we give the definitions of the functional tetrahedron and the parametric tetrahedron equation, and we explain their relation with matrix refactorisation problems. Furthermore, we prove a statement regarding the invariants of the solutions to these equations which indicate the integrability of these solutions. 

In section 3, we present the basic steps of a simple scheme for constructing solutions to the parametric tetrahedron equation. Also, we list all the Darboux transformations that we use throughout the text, without giving details on their derivation. For their construction one can refer to \cite{SPS}.

In section 4, we study matrix refactorisation problems for Darboux matrices associated with the NLS equation, and we derive novel solutions to the tetrahedron equation and the parametric tetrahedron equation. In particular, using an NLS type Darboux transformation, we derive a novel nine-dimensional birational tetrahedron map which can be restricted to a novel six-dimensional, birational, parametric tetrahedron map on invariant leaves. Then, considering the matrix refactorisation problem for a degenerated Darboux matrix of NLS type, we construct another birational, six-dimensional, parametric tetrahedron map which can be restricted to a three-dimensional one on the level sets of its invariants. The latter map at a certain limit gives a map from Sergeev's classification \cite{Sergeev}.

In section 5, we employ Darboux transformations associated with the DNLS equation in order to derive novel solutions to the tetrahedron equation. We first consider a Darboux transformation associated with the DNLS equation, and we obtain a novel nine-dimensional  tetrahedron map the restriction of which on invariant leaves is represented by a novel six-dimensional parametric tetrahedron map. Moreover, we construct another six-dimensional parametric tetrahedron map using a degenerated version of the Darboux transformation for the DNLS equation. The derived six-dimensional map can also be restricted to a three-dimensional one which gives a map from Sergeev's classification at a certain limit. All the maps derived in this section are noninvolutive and birational.

Finally, in section 6, we present a brief summary of the results of the paper and discuss some ideas for potential extensions of the obtained results.

\section{Preliminaries}\label{preliminaries}
\subsection{Tetrahedron maps}
Let $\mathcal{X}$ be an algebraic variety in $\mathbb{C}^N$. A map $T\in\End(\mathcal{X}^3)$, namely
\begin{equation}\label{Tetrahedron_map}
 T:(x,y,z)\mapsto (u(x,y,z),v(x,y,z),w(x,y,z)),
\end{equation}
is called a \textit{tetrahedron map} if it satisfies the \textit{functional tetrahedron} (or Zamolodchikov) equation
\begin{equation}\label{Tetrahedron-eq}
    T^{123}\circ T^{145} \circ T^{246}\circ T^{356}=T^{356}\circ T^{246}\circ T^{145}\circ T^{123}.
\end{equation}
Functions $T^{ijk}\in\End(\mathcal{X}^6)$, $i,j=1,2,3,~i\neq j$, in \eqref{Tetrahedron-eq} are maps that act as map $T$ on the $ijk$ terms of the Cartesian product $\mathcal{X}^6$ and trivially on the others. For instance,
$$
T^{145}(x,y,z,r,s,t)=(u(x,r,s),y,z,v(x,r,s),w(x,r,s),t).
$$

A tetrahedron map can be represented on the faces of the cube, as in Figure \ref{Tet_map}, mapping three neighbour faces of the cube to the rest three on the opposite side.  In particular, a tetrahedron map can be thought as a map mapping the values $x$,  $y$ and $z$, assigned to the back, bottom and left sides of the cube (i.e. faces 1, 2 and 3 in the left half-cube of Figure \ref{Tet_map}), to values $u$, $v$ and $w$ assigned to the front, top and right faces of the cube, respectively.
\begin{figure}[ht]
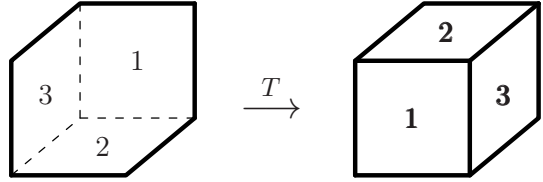

\centering
\centertexdraw{ 
\setunitscale .6
\move(1.6 1.5) \setgray 0 \linewd 0.04 \lvec (0.6 1.5) \lvec(0 1) \lvec(0 0)
\move (0 0) \linewd 0.04 \lvec (1 0) \lvec(1.6 .5) \linewd 0.01  \lpatt(0.067 0.09) \lvec(0.6 .5) \lvec(0 0) 
\lpatt() \move(1.6 .5)\linewd 0.04 \lvec(1.6 1.5) 
 \linewd 0.01 \lpatt(0.067 0.09)  \move(0.6 1.5)\lvec(0.6 .5) \lpatt()
\move(3 0)\setgray 0 \linewd 0.04 \lvec(3 1)\lvec(3.6 1.5)  \lvec(4.6 1.5)\setgray 0 \linewd 0.04
\move (3 0) \linewd 0.04 \lvec (4 0) \lvec(4.6 .5) \lvec(4.6 1.5) \linewd 0.04 \lvec(4 1) \lvec(4 0) 
\move(3 1)\lvec(4 1)

\htext (2 .5) {{\Large $\overset{T}{\longrightarrow}$}}
\textref h:C v:C \htext(.8 .25){2}
\textref h:C v:C \htext(.3 .7){3}
\textref h:C v:C \htext(1.1 1){1}

\textref h:C v:C \htext(3.8 1.25){{\pmb 2}}
\textref h:C v:C \htext(4.3 .7){{\pmb 3}}
\textref h:C v:C \htext(3.5 .5){{\pmb 1}}

}
\caption{Tetrahedron map. Schematic representation.}\label{Tet_map}
\end{figure}

The most, probably, celebrated Tetrahedron map is \cite{Kashaev-Sergeev}
$$
(x,y,z)\overset{T}{\rightarrow}\left(\frac{xy}{x+z+xyz},x+z+xyz,\frac{yz}{x+z+xyz}\right),
$$
following from the `star-triangle' transformation in electric circuits.

Now, if we assign the complex parameters $a$, $b$ and $c$ to the variables $x$, $y$ and $z$, respectively, we define a map $T\in\End[(\mathcal{X}\times\mathbb{C})^3]$, namely $T:((x,a),(y,b),(z,c))\mapsto ((u(x,y,z),a),(v(x,y,z),b),(w(x,y,z),c))$ which we denote for simplicity as
\begin{equation}\label{Par-Tetrahedron_map}
 T_{a,b,c}:(x,y,z)\mapsto (u_{a,b,c}(x,y,z),v_{a,b,c}(x,y,z),w_{a,b,c}(x,y,z)).
\end{equation}
Map \eqref{Par-Tetrahedron_map} is called a \textit{parametric tetrahedron map} if it satisfies the \textit{parametric functional tetrahedron equation}
\begin{equation}\label{Par-Tetrahedron-eq}
    T^{123}_{a,b,c}\circ T^{145}_{a,d,e} \circ T^{246}_{b,d,f}\circ T^{356}_{c,e,f}=T^{356}_{c,e,f}\circ T^{246}_{b,d,f}\circ T^{145}_{a,d,e}\circ T^{123}_{a,b,c}.
\end{equation}
Schematically, a parametric tetrahedron map can be understood as in Figure \ref{Tet_map} where the parameters $a$, $b$ and $c$ are placed on the faces of the cube together with the variables $x$, $y$ and $z$, respectively.  Moreover, a four-dimensional representation of the parametric tetrahedron equation can be seen in Figure \ref{Tet_eq-rep}. 

\begin{figure}[ht]
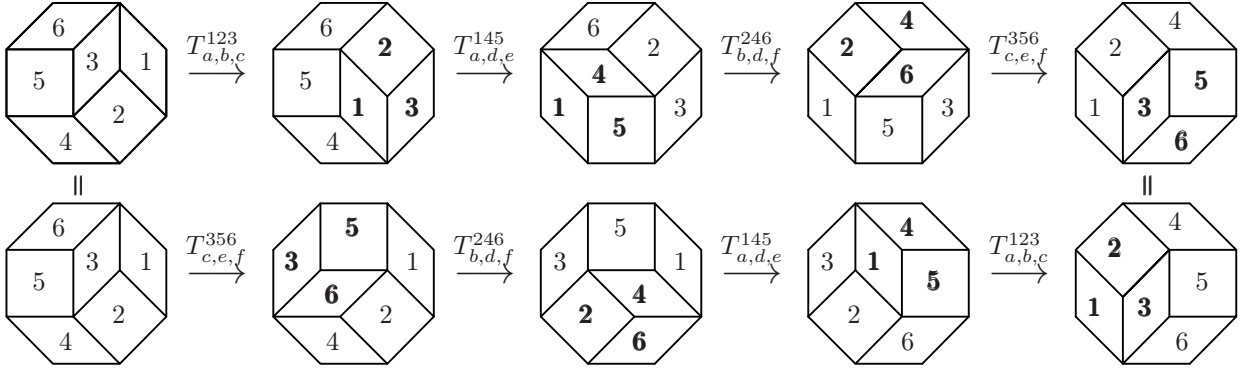

\centering
\centertexdraw{ 
\setunitscale .35

\move (0 0) \lvec (0 -1)\lvec(0.7 -1.7)\lvec(0.7 -0.7)\lvec(0 0)
\move (0 0) \lvec (-1 0)
\move (0 0)\lvec(-0.7 -0.7)\lvec(-0.7 -1.7)\lvec(0 -1)
\move (-0.7 -0.7)\lvec(-1.7 -0.7)\lvec(-1 0)
\move(-0.7 -1.7)\lvec(0 -2.4)\lvec(0.7 -1.7)
\move(-0.7 -1.7)\lvec(-1.7 -1.7)\lvec(-1.7 -0.7)
\move(0 -2.4)\lvec(-1 -2.4)\lvec(-1.7 -1.7)

\htext (-0.5 -1){{\small 3}}
\htext (0.3 -1){{\small 1}}
\htext (-0.1 -1.8){{\small 2}}
\htext (-.9 -2.2){{\small 4}}
\htext (-1.3 -1.3){{\small 5}}
\htext (-1 -.5){{\small 6}}

\htext (1 -1.2) {{\Large $\overset{T^{123}_{a,b,c}}{\longrightarrow}$}}

\move (4 0) \lvec(4.7 -0.7)\lvec(4.7 -1.7)
\move (4 0) \lvec (3 0)
\move (4 0)\lvec(3.3 -0.7)\lvec(3.3 -1.7)
\move(3.3 -0.7)\lvec(4 -1.4)\lvec(4.7 -0.7)
\move(4 -1.4)\lvec(4 -2.4)
\move (3.3 -0.7)\lvec(2.3 -0.7)\lvec(3 0)
\move(3.3 -1.7)\lvec(4 -2.4)\lvec(4.7 -1.7)
\move(3.3 -1.7)\lvec(2.3 -1.7)\lvec(2.3 -0.7)
\move(4 -2.4)\lvec(3 -2.4)\lvec(2.3 -1.7)

\htext (3.9 -0.8){{\small {\pmb 2}}}
\htext (3.5 -1.7){{\small {\pmb 1}}}
\htext (4.3 -1.7){{\small {\pmb 3}}}
\htext (3.1 -2.2){{\small 4}}
\htext (2.7 -1.3){{\small 5}}
\htext (3 -.5){{\small 6}}

\htext (5 -1.2) {{\Large $\overset{T^{145}_{a,d,e}}{\longrightarrow}$}}

\move (8 0) \lvec(8.7 -0.7)\lvec(8.7 -1.7)
\move (8 0) \lvec (7 0)
\move (8 0)\lvec(7.3 -0.7)
\move(7.3 -0.7)\lvec(8 -1.4)\lvec(8.7 -0.7)
\move(8 -1.4)\lvec(7 -1.4)\lvec(7 -2.4)
\move(7 -1.4)\lvec(6.3 -0.7)
\move(8 -1.4)\lvec(8 -2.4)
\move (7.3 -0.7)\lvec(6.3 -0.7)\lvec(7 0)
\move(8 -2.4)\lvec(8.7 -1.7)
\move(6.3 -1.7)\lvec(6.3 -0.7)
\move(8 -2.4)\lvec(7 -2.4)\lvec(6.3 -1.7)

\htext (7.9 -0.8){{\small 2}}
\htext (6.5 -1.7){{\small {\pmb 1}}}
\htext (8.3 -1.7){{\small 3}}
\htext (7.1 -1.2){{\small {\pmb 4}}}
\htext (7.4 -2){{\small {\pmb 5}}}
\htext (7 -.5){{\small 6}}

\htext (9 -1.2) {{\Large $\overset{T^{246}_{b,d,f}}{\longrightarrow}$}}

\move (12 0) \lvec(12.7 -0.7)\lvec(12.7 -1.7)
\move (12 0) \lvec (11 0)
\move(12 -1.4)\lvec(12.7 -0.7)\lvec(11.7 -0.7)\lvec(11 -1.4)
\lvec(11.7 -0.7)\lvec(11 0)
\move(12 -1.4)\lvec(11 -1.4)\lvec(11 -2.4)
\move(11 -1.4)\lvec(10.3 -0.7)
\move(12 -1.4)\lvec(12 -2.4)
\move(10.3 -0.7)\lvec(11 0)
\move(12 -2.4)\lvec(12.7 -1.7)
\move(10.3 -1.7)\lvec(10.3 -0.7)
\move(12 -2.4)\lvec(11 -2.4)\lvec(10.3 -1.7)

\htext (11.7 -0.4){{\small {\pmb 4}}}
\htext (10.5 -1.7){{\small 1}}
\htext (12.3 -1.7){{\small 3}}
\htext (11.7 -1.2){{\small {\pmb 6}}}
\htext (11.4 -2){{\small 5}}
\htext (10.8 -.8){{\small {\pmb 2}}}

\htext (13 -1.2) {{\Large $\overset{T^{356}_{c,e,f}}{\longrightarrow}$}}

\move (16 0) \lvec(16.7 -0.7)\lvec(16.7 -1.7)\lvec(15.7 -1.7)\lvec(15 -2.4)
\move(15.7 -1.7)\lvec(15.7 -0.7)
\move (16 0) \lvec (15 0)
\move(16.7 -0.7)\lvec(15.7 -0.7)\lvec(15 -1.4)
\lvec(15.7 -0.7)\lvec(15 0)
\move(15 -1.4)\lvec(15 -2.4)
\move(15 -1.4)\lvec(14.3 -0.7)
\move(14.3 -0.7)\lvec(15 0)
\move(16 -2.4)\lvec(16.7 -1.7)
\move(14.3 -1.7)\lvec(14.3 -0.7)
\move(16 -2.4)\lvec(15 -2.4)\lvec(14.3 -1.7)

\htext (15.7 -0.4){{\small 4}}
\htext (14.5 -1.7){{\small 1}}
\htext (16.1 -1.3){{\small {\pmb 5}}}
\htext (15.25 -1.7){{\small {\pmb 3}}}
\htext (15.8 -2.2){{\small {\pmb 6}}}
\htext (14.8 -.8){{\small 2}}

\move (0 0) \lvec (0 -1)\lvec(0.7 -1.7)\lvec(0.7 -0.7)\lvec(0 0)
\move (0 0) \lvec (-1 0)
\move (0 0)\lvec(-0.7 -0.7)\lvec(-0.7 -1.7)\lvec(0 -1)
\move (-0.7 -0.7)\lvec(-1.7 -0.7)\lvec(-1 0)
\move(-0.7 -1.7)\lvec(0 -2.4)\lvec(0.7 -1.7)
\move(-0.7 -1.7)\lvec(-1.7 -1.7)\lvec(-1.7 -0.7)
\move(0 -2.4)\lvec(-1 -2.4)\lvec(-1.7 -1.7)

\move (0 -3) \lvec (0 -4)\lvec(0.7 -4.7)\lvec(0.7 -3.7)\lvec(0 -3)
\move (0 -3) \lvec (-1 -3)
\move (0 -3)\lvec(-0.7 -3.7)\lvec(-0.7 -4.7)\lvec(0 -4)
\move (-0.7 -3.7)\lvec(-1.7 -3.7)\lvec(-1 -3)
\move(-0.7 -4.7)\lvec(0 -5.4)\lvec(0.7 -4.7)
\move(-0.7 -4.7)\lvec(-1.7 -4.7)\lvec(-1.7 -3.7)
\move(0 -5.4)\lvec(-1 -5.4)\lvec(-1.7 -4.7)

\htext (-0.5 -4){{\small 3}}
\htext (0.3 -4){{\small 1}}
\htext (-0.1 -4.8){{\small 2}}
\htext (-.9 -5.2){{\small 4}}
\htext (-1.3 -4.3){{\small 5}}
\htext (-1 -3.5){{\small 6}}

\htext (1 -4.2) {{\Large $\overset{T^{356}_{c,e,f}}{\longrightarrow}$}}

\move (4 -3) \lvec (4 -4)\lvec(4.7 -4.7)\lvec(4.7 -3.7)\lvec(4 -3)
\move (4 -3) \lvec (3 -3)
\move (3.3 -4.7)\lvec(4 -4)\lvec(3 -4)\lvec(3 -3)
\move (3 -4)\lvec(2.3 -4.7)
\move (2.3 -3.7)\lvec(3 -3)
\move(3.3 -4.7)\lvec(4 -5.4)\lvec(4.7 -4.7)
\move(3.3 -4.7)\lvec(2.3 -4.7)\lvec(2.3 -3.7)
\move(4 -5.4)\lvec(3 -5.4)\lvec(2.3 -4.7)

\htext (3.1 -4.5){{\small {\pmb 6}}}
\htext (4.3 -4){{\small 1}}
\htext (3.9 -4.8){{\small 2}}
\htext (3.1 -5.2){{\small 4}}
\htext (2.5 -4){{\small {\pmb 3}}}
\htext (3.4 -3.5){{\small {\pmb 5}}}

\htext (5 -4.2) {{\Large $\overset{T^{246}_{b,d,f}}{\longrightarrow}$}}

\move (8 -3) \lvec (8 -4)\lvec(8.7 -4.7)\lvec(8.7 -3.7)\lvec(8 -3)
\move (8 -3) \lvec (7 -3)
\move (8 -4)\lvec(7 -4)\lvec(7 -3)
\move (7 -4)\lvec(6.3 -4.7)
\move (6.3 -3.7)\lvec(7 -3)
\move(8 -5.4)\lvec(8.7 -4.7)
\move(6.3 -4.7)\lvec(6.3 -3.7)
\move(8 -5.4)\lvec(7 -5.4)\lvec(6.3 -4.7)
\move (8.7 -4.7)\lvec(7.7 -4.7)\lvec (7 -5.4)
\move (7.7 -4.7)\lvec(7 -4)

\htext (6.9 -4.8){{\small {\pmb 2}}}
\htext (8.3 -4){{\small 1}}
\htext (7.7 -4.5){{\small {\pmb 4}}} 
\htext (7.7 -5.2){{\small {\pmb 6}}} 
\htext (6.5 -4){{\small 3}}
\htext (7.4 -3.5){{\small 5}}

\htext (9 -4.2) {{\Large $\overset{T^{145}_{a,d,e}}{\longrightarrow}$}}

\move (12.7 -4.7)\lvec(12.7 -3.7)\lvec(12 -3)
\move(12.7 -3.7)\lvec(11.7 -3.7)\lvec(11.7 -4.7)
\move(11.7 -3.7)\lvec(11 -3)
\move (12 -3) \lvec (11 -3)
\move(11 -4)\lvec(11 -3)
\move (11 -4)\lvec(10.3 -4.7)
\move (10.3 -3.7)\lvec(11 -3)
\move(12 -5.4)\lvec(12.7 -4.7)
\move(10.3 -4.7)\lvec(10.3 -3.7)
\move(12 -5.4)\lvec(11 -5.4)\lvec(10.3 -4.7)
\move (12.7 -4.7)\lvec(11.7 -4.7)\lvec (11 -5.4)
\move (11.7 -4.7)\lvec(11 -4)

\htext (10.9 -4.8){{\small 2}}
\htext (12.1 -4.3){{\small {\pmb 5}}} 
\htext (11.2 -4){{\small {\pmb 1}}} 
\htext (11.7 -5.2){{\small 6}} 
\htext (10.5 -4){{\small 3}}
\htext (11.7 -3.5){{\small {\pmb 4}}}

\htext (13 -4.2) {{\Large $\overset{T^{123}_{a,b,c}}{\longrightarrow}$}}

\move (16 -3) \lvec(16.7 -3.7)\lvec(16.7 -4.7)\lvec(15.7 -4.7)\lvec(15 -5.4)
\move(15.7 -4.7)\lvec(15.7 -3.7)
\move (16 -3) \lvec (15 -3)
\move(16.7 -3.7)\lvec(15.7 -3.7)\lvec(15 -4.4)
\lvec(15.7 -3.7)\lvec(15 -3)
\move(15 -4.4)\lvec(15 -5.4)
\move(15 -4.4)\lvec(14.3 -3.7)
\move(14.3 -3.7)\lvec(15 -3)
\move(16 -5.4)\lvec(16.7 -4.7)
\move(14.3 -4.7)\lvec(14.3 -3.7)
\move(16 -5.4)\lvec(15 -5.4)\lvec(14.3 -4.7)

\htext (15.7 -3.4){{\small 4}}
\htext (14.5 -4.7){{\small {\pmb 1}}}
\htext (16.1 -4.3){{\small 5}}
\htext (15.25 -4.7){{\small {\pmb 3}}}
\htext (15.8 -5.2){{\small 6}}
\htext (14.8 -3.8){{\small {\pmb 2}}}

\vtext (15.5 -2.9){$\pmb{=}$}
\vtext (-.5 -2.9){$\pmb{=}$}
}
\caption{Tetrahedron equation. Schematic representation (\cite{Doliwa-Kashaev}).}
\label{Tet_eq-rep}
\end{figure}

Only a handful of solutions to the parametric tetrahedron equation \eqref{Par-Tetrahedron-eq} are known to date; an interesting parametric tetrahedron map can be found in \cite{Bazhanov-Sergeev}.

\subsection{Tetrahedron maps and matrix refactorisation problems}
Let $L=L(x,a;\lambda)$ be a matrix depending on a variable $x\in\mathcal{X}$, a parameter $a\in\mathbb{C}$ and a spectral parameter $\lambda\in\mathbb{C}$ of the form
\begin{equation}\label{matrix-L}
   L(x,a;\lambda)= \begin{pmatrix} 
A(x,a;\lambda) & B(x,a;\lambda)\\ 
C(x,a;\lambda) & D(x,a;\lambda)
\end{pmatrix},
\end{equation}
where its entries $A,B,C,D$ are scalar functions of $x$, $a$ and $\lambda$. Moreover, we define the following matrices
{\small
\begin{equation}\label{Lij-mat}
   L_{12}=\begin{pmatrix} 
 A(x,a;\lambda) &  B(x,a;\lambda) & 0\\ 
C(x,a;\lambda) &  D(x,a;\lambda) & 0\\
0 & 0 & 1
\end{pmatrix},\quad
 L_{13}= \begin{pmatrix} 
 A(x,a;\lambda) & 0 & B(x,a;\lambda)\\ 
0 & 1 & 0\\
C(x,a;\lambda) & 0 & D(x,a;\lambda)
\end{pmatrix}, \quad
 L_{23}=\begin{pmatrix} 
   1 & 0 & 0 \\
0 & A(x,a;\lambda) & B(x,a;\lambda)\\ 
0 & C(x,a;\lambda) & D(x,a;\lambda)
\end{pmatrix},
\end{equation}
}
where $L_{ij}=L_{ij}(x,a;\lambda)$, $i,j=1,2,3$.

Following the work of Sergeev \cite{Sergeev} and the work of Kashaev, Korepanov and Sergeev \cite{Kashaev-Sergeev} we study the solutions of the following matrix trifactorisation problem
\begin{equation}\label{Lax-Tetra}
    L_{12}(u,a;\lambda)L_{13}(v,b;\lambda)L_{23}(w,c;\lambda)= L_{23}(z,c;\lambda)L_{13}(y,b;\lambda)L_{12}(x,a;\lambda).
\end{equation}
If the above matrix trifactorisation problem defines a map, we will call equation \eqref{Lax-Tetra} its \textit{Lax representation}. Equation \eqref{Lax-Tetra} was also studied by Korepanov in \cite{Korepanov} in the more general case where functions $A, B, C$ and $D$ in \eqref{Lij-mat} are matrices.

Unlike the case of matrix refactorisation problems associated with solutions to the Yang--Baxter equation \cite{Veselov2}, the matrix rafactorisation problem \eqref{Lax-Tetra} does not admit the symmetry $(u,v,w;a,b,c)\rightarrow (x,y,z;c,b,a)$. That means that the rationality of a map $(x,y,z)\rightarrow(u(x,y;a,b),v(x,y;a,b),w(x,y;a,b))$ defined by \eqref{Lax-Tetra} does not necessarily imply its birationality. Moreover, the trace of the right-hand side of \eqref{Lax-Tetra} does not necessarily generate invariants for the former map. However, we have the following.

\begin{proposition}\label{symm-Map-inv}
If $L=L(x,a;\lambda)$ is a matrix of the form \eqref{matrix-L} with $B(x,a;\lambda)=C(x,a;\lambda)$, then the quantity
$$
\tr\big(L_{23}(z,c;\lambda)L_{13}(y,b;\lambda)L_{12}(x,a;\lambda)\big)
$$ 
is a generator of invariants of the map $(x,y,z)\rightarrow(u(x,y;a,b),v(x,y;a,b),w(x,y;a,b))$ defined by \eqref{Lax-Tetra}.
\end{proposition}
\begin{proof}
Indeed, if $B(x,a;\lambda)=C(x,a;\lambda)$, then the matrices $L_{ij}$, $i,j=1,2,3$, $i<j$, are symmetric, namely $L_{ij}=L_{ij}^T$. Therefore, the trace of the left-hand side of \eqref{Lax-Tetra}:
\begin{align}\label{trace-prop}
\tr\big(L_{12}(u,a;\lambda)L_{13}(v,b;\lambda)L_{23}(w,c;\lambda)\big)&\overset{\eqref{Lax-Tetra}}{=}\tr\big(L_{23}(z,c;\lambda)L_{13}(y,b;\lambda)L_{12}(x,a;\lambda)\big)\nonumber\\
& = \tr\big((L_{23}(z,c;\lambda)L_{13}(y,b;\lambda)L_{12}(x,a;\lambda))^T\big)\nonumber\\
& =  \tr\big(L_{12}^T(x,a;\lambda)L_{13}^T(y,b;\lambda)L_{23}^T(z,c;\lambda)\big)\nonumber\\
& =\tr\big(L_{12}(x,a;\lambda)L_{13}(y,b;\lambda)L_{23}(z,c;\lambda)\big),
\end{align}
since $L_{ij}=L_{ij}^T$, $L_{ij}$, $i,j=1,2,3$, $i<j$.

Now, if we expand in $\lambda$: $\tr\big(L_{23}(z,c;\lambda)L_{13}(y,b;\lambda)L_{12}(x,a;\lambda)\big)=\sum_{k}I_k(x,y,z)\lambda^k$, then from \eqref{trace-prop} follows that
$$
I_k(x,y,z)=I_k(u,v,w),
$$ 
i.e. $I_k(x,y,z)$ are invariants of the map $(x,y,z)\rightarrow(u(x,y;a,b),v(x,y;a,b),w(x,y;a,b))$ defined by \eqref{Lax-Tetra}.
\end{proof}

\section{Derivation of parametric tetrahedron maps: a Darboux construction scheme}
In this section, following the ideas presented in \cite{Sokor-Sasha} for the case of parametric Yang--Baxter maps, we demonstrate a scheme for constructing solutions to the parametric tetrahedron equation using Darboux transformations. The scheme can be summarised in Figure \ref{Darboux scheme}. 
\begin{figure}[ht]
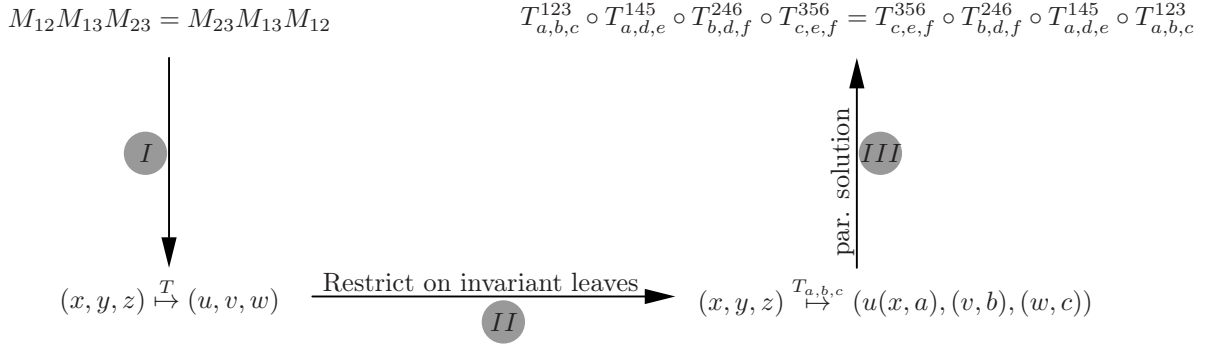

\centering
\centertexdraw{ 
\setunitscale 0.5
\move (-3.2 -.7)  \arrowheadtype t:F \avec(.6 -.7)
\move (-4.7 1.8)  \arrowheadtype t:F \avec(-4.7 -.4)
\move (2.5 -.4)  \arrowheadtype t:F \avec(2.5 1.8) 
\textref h:C v:C \htext(-4.7 2.2){\small $M_{12}M_{13}M_{23}=M_{23}M_{13}M_{12}$}
\textref h:C v:C \htext(2.5 2.2){\small $T^{123}_{a,b,c}\circ T^{145}_{a,d,e} \circ T^{246}_{b,d,f}\circ T^{356}_{c,e,f}=T^{356}_{c,e,f}\circ T^{246}_{b,d,f}\circ T^{145}_{a,d,e}\circ T^{123}_{a,b,c}$}
\textref h:C v:C \htext(-4.7 -.7){\small $(x,y,z)\stackrel{T}{\mapsto}(u,v,w)$}
\textref h:C v:C \htext(2.9 -.7){\small $(x,y,z)\stackrel{T_{a,b,c}}{\mapsto}(u(x,a),(v,b),(w,c))$}\lpatt(0.067 0.1)
\move (-4.95 .8)\fcir f:.6 r:.23
\textref h:C v:C \small{\htext(-4.95 .8){$I$}}
\move (2.75 .8)\fcir f:.6 r:.23
\textref h:C v:C \small{\htext(2.75 .8){$III$}}
\move (-1.2 -.97)\fcir f:.6 r:.23
\textref h:C v:C \small{\htext(-1.2 -.97){$II$}}
\textref h:C v:C \small{\htext(-1.45 -.55){\small Restrict on invariant leaves}}
\textref h:C v:C \small{\vtext(2.35 .5){\small par. solution}}
}
\caption{Darboux construction scheme.}\label{Darboux scheme}
\end{figure}

\textbf{I.} Consider matrix trifactorisation problems \eqref{Lax-Tetra} for Darboux matrices associated with Lax operators of the form $\mathcal{L}=D_x+U(\mathbf{u}(x,t),\lambda)$, where $U(\mathbf{u}(x,t),\lambda)$ belongs to the Lie algebra $\mathfrak{sl}_2$ and depends on a (vector) potential function $\mathbf{u}$, a spectral parameter $\lambda$ and two variables $x$ and  $t$, implicitly through the potential function. From such matrix refactorisation problems we derive tetrahedron maps.

\textbf{II.} For such Darboux transformations as in step \textbf{I}, the associated B\"acklund transformations admit first integrals (see \cite{SPS}). These first integrals indicate the existence of inviariants of the tetrahedron maps derived in step \textbf{I}.  On the level sets of these integrals, the tetrahedron maps derived in step \textbf{I} are represented as parametric maps.

\textbf{III.} These parametric maps are solutions to the parametric tetrahedron equation.

In the next section, we apply the above scheme to particular Darboux transformations related to the NLS and the DNLS equation.

\subsection{Nonlinear Schr\"odinger type Darboux transformations}
Here, we list all the Darboux transformations that we use in this text. For details regarding their derivation, one can refer to \cite{SPS}. In particular, we use Darboux transformations associated with the Lax operator $\mathcal{L}_{NLS} \in\mathfrak{sl}_2\left[\lambda\right]\left[D_x\right]$ and $\mathcal{L}_{DNLS}\in\mathfrak{sl}_2[\lambda]^{\langle s_1 \rangle}[D_x]$, namely
\begin{equation}\label{operators}
\mathcal{L}_{NLS} =D_x+\lambda \left(\begin{array}{cc}1 & 0\\0 & -1\end{array}\right)+\left(\begin{array}{cc}0 & 2p\\2q & 0\end{array}\right),\quad
\mathcal{L}_{DNLS}=D_x+\lambda^2 \left(\begin{array}{cc}1 & 0\\0 & -1\end{array}\right)+\lambda\left(\begin{array}{cc}0 & 2p\\2q & 0\end{array}\right)
\end{equation}
i.e. the spatial parts of the Lax pairs for the NLS and the DNLS equation, respectively.

\begin{enumerate}
    \item A Darboux transformation for $\mathcal{L}_{NLS}$ in \eqref{operators} is:
\begin{equation}\label{DM-NLS}
  M=\lambda \left(
     \begin{array}{cc}
         1 & 0\\
         0 & 0
     \end{array}\right)+\left(
     \begin{array}{cc}
         f & p\\
         \tilde{q} & 1
     \end{array}\right),
\end{equation}
where its entries obey the system of equations
$\partial_x f=2 (pq-\tilde{p}\tilde{q}), \quad \partial_x p =2 (p f -\tilde{p}),\quad \partial_x \tilde{q}=2 ( q-\tilde{q}f)$, i.e. the so-called B\"acklund transformation. A first integral of this  system of differential equations is
\begin{equation}\label{1stInt}
\partial_x (f-p\tilde{q})=0.
\end{equation}
\item A `degenerated' Darboux Matrix for $\mathcal{L}_{NLS}$ reads
\begin{equation}\label{M-degen}
M_a(p,f)=\lambda \left(\begin{array}{cc}1 & 0\\0 & 0\end{array}\right)+\left(\begin{array}{cc}f & p\\\frac{a}{p} & 0 \end{array}\right),\quad f=\frac{p_x}{2p}.
\end{equation}
\item A Darboux transformation associated with operator $\mathcal{L}_{DNLS}$ in \eqref{operators} is
\begin{equation} \label{DT-sl2-gen}
M(p,\tilde{q},f) := \lambda^{2}\left(\begin{array}{cc} f & 0\\ 0 & 0\end{array}\right)+\lambda\left(\begin{array}{cc} 0 & fp\\ f\tilde{q} & 0\end{array}\right)+\left(\begin{array}{cc} 0 & 0\\ 0 & 1 \end{array}\right),
\end{equation}
where $p$ and $q$ satisfy a system of differential equations which possesses the following first integral,
\begin{equation} \label{sl2-D-con-det-gen}
\partial_x \left(f^{2}p\tilde{q}- f\right)=0.
\end{equation}
\item A `degenerated' Darboux transformation related to $\mathcal{L}_{DNLS}$ is
\begin{equation} \label{DT-sl2-degen}
M(p,f;a) := \lambda^{2}\left(\begin{array}{cc} f & 0\\ 0 & 0\end{array}\right)+\lambda\left(\begin{array}{cc} 0 & fp\\ \frac{a}{fp} & 0\end{array}\right)+\left(\begin{array}{cc} 1 & 0\\ 0 & 0 \end{array}\right).
\end{equation}
\end{enumerate}

In the following sections, we study matrix trifactorsation problems associated with the above Darboux matrices in order to derive tetrahedron maps. As we shall see, the existence of first integrals \eqref{1stInt} and \eqref{sl2-D-con-det-gen} play an important role, since they indicate the invariant leaves on which the derived tetrahedron maps are expressed as parametric tetrahedron maps, namely solutions to \eqref{Par-Tetrahedron-eq}.

\section{Nonlinear Schr\"odinger type tetrahedron maps}
In this section, we study matrix trifactorisation problems \eqref{Lax-Tetra} for the Darboux transformations \eqref{DM-NLS} and \eqref{M-degen} associated to the NLS equation, and we construct novel solutions to the tetrahedron equation and the parametric tetrahedron equation. All the derived maps are noninvolutive which are in general more interesting comparing to involutive ones, since involutive maps have trivial dynamics. 

\subsection{A novel nine-dimensional tetrahedron map}
Changing $(p,\tilde{q},f+\lambda)\rightarrow (x_1,x_2,X)$ in \eqref{DM-NLS}, we define the following matrix
\begin{equation}\label{M-NLS}
  M(x_1,x_2,X)=\left(
     \begin{array}{cc}
         X & x_1\\
         x_2 & 1
     \end{array}\right).
\end{equation}
For this matrix \eqref{M-NLS} we consider the matrix trifactorisation problem \eqref{Lax-Tetra}, namely
\begin{equation}\label{NLS-Lax-Tetra}
    M_{12}(u_1,u_2,U)M_{13}(v_1,v_2,V)M_{23}(w_1,w_2,W)= M_{23}(z_1,z_2,Z)M_{13}(y_1,y_2,Y)M_{12}(x_1,x_2,X).
\end{equation}
The above matrix refactorisation problem implies the following:
\begin{subequations}\label{NLS-correspondence}
\begin{align}
    u_1&=\frac{x_1(y_1y_2-Y)+y_1z_2}{z_1z_2-Z},\quad u_2=\frac{x_2Z+y_2z_1X}{XY}U\label{NLS-correspondence-a}\\
    v_1&=\frac{y_1Z+x_1z_1(y_1y_2-Y)}{y_1y_2z_1(x_1y_2+z_2)X-(x_1y_2z_1+z_1z_2-Z)XY+x_2[y_1z_2+x_1(y_1y_2-Y)]Z}\frac{XY}{U}\label{NLS-correspondence-b}\\
    v_2&=x_2z_2+y_2X,\quad V=\frac{XY}{U\label{NLS-correspondence-c}}\\
    w_1&=\frac{[x_2y_1Z+(y_1y_2-Y)z_1X](z_1z_2-Z)}{y_1y_2z_1(x_1y_2+z_2)X-(x_1y_2z_1+z_1z_2-Z)XY+x_2[y_1z_2+x_1(y_1y_2-Y)]Z}\label{NLS-correspondence-d}\\
    w_2&=x_1y_2+z_2,\label{NLS-correspondence-e}\\ W&=\frac{(x_1x_2-X)(z_1z_2-Z)YZ}{y_1y_2z_1(x_1y_2+z_2)X-(x_1y_2z_1+z_1z_2-Z)XY+x_2[y_1z_2+x_1(y_1y_2-Y)]Z}\label{NLS-correspondence-f}
\end{align}
\end{subequations}
which is a correspondence rather than a map; functions $u_2,v_1$ and $V$ are defined in terms of $U$. The above correspondence does not satisfy the tetrahedron equation for any choice of $U$. However, as we shall see below, there is at least a choice of $U$ for which \eqref{NLS-correspondence} defines a tetrahedron map.

In particular, the determinant of equation \eqref{NLS-Lax-Tetra} implies the equation 
$$
(U-u_1u_2)(V-v_1v_2)(W-w_1w_2)=(X_1-x_1x_2)(Y-y_1y_2)(Z-z_1z_2).
$$
We choose $U-u_1u_2=X-x_1x_2$, $V-v_1v_2=Y-y_1y_2$, $W-w_1w_2=Z-z_1z_2$. Then, the following holds.
\begin{proposition}
The system consisting of equation \eqref{Lax-Tetra} together with $U-u_1u_2=X-x_1x_2$  has a unique solution, namely a map 
$
(x_1,x_2,X,y_1,y_2,Y,z_1,z_2,Z)\overset{T}{\longrightarrow} (u_1,u_2,U,v_1,v_2,V,w_1,w_2,W),
$
given by
\begin{subequations}\label{Tet-NLS-9D} 
\begin{align}
x_1\mapsto u_1 &=\frac{x_1(y_1y_2-Y)+y_1z_2}{z_1z_2-Z},\label{Tet-NLS-9D-a}\\
x_2\mapsto u_2 &=\frac{(x_1x_2-X)(y_2z_1X+x_2Z)(z_1z_2-Z)}{y_1y_2z_1(x_1y_2+z_2)X-(x_1y_2z_1+z_1z_2-Z)XY+x_2[y_1z_2+x_1(y_1y_2-Y)]Z},\label{Tet-NLS-9D-b}\\
X\mapsto U &=\frac{(x_1x_2-X)(y_1y_2-Y)X}{y_1y_2z_1(x_1y_2+z_2)X-(x_1y_2z_1+z_1z_2-Z)XY+x_2[y_1z_2+x_1(y_1y_2-Y)]Z},\label{Tet-NLS-9D-c}\\
y_1\mapsto v_1 &=\frac{x_1z_1(y_1y_2-Y)+y_1Z}{(x_1x_2-X)(z_1z_2-Z)},\label{Tet-NLS-9D-d}\\
y_2\mapsto v_2 &=x_2z_2+y_2X,\label{Tet-NLS-9D-e}\\
Y\mapsto V &=\frac{y_1y_2z_1(x_1y_2+z_2)X-(x_1y_2z_1+z_1z_2-Z)XY+x_2[y_1z_2+x_1(y_1y_2-Y)]Z}{(x_1x_2-X)(z_1z_2-Z)},\label{Tet-NLS-9D-f}\\
z_1\mapsto w_1 &=\frac{[x_2y_1Z-z_1(y_1y_2-Y)X](z_1z_2-Z)}{y_1y_2z_1(x_1y_2+z_2)X-(x_1y_2z_1+z_1z_2-Z)XY+x_2[y_1z_2+x_1(y_1y_2-Y)]Z},\label{Tet-NLS-9D-g} \\
z_2\mapsto w_2 &=x_1y_2+z_2,\label{Tet-NLS-9D-h}\\
Z\mapsto W &=\frac{(x_1x_2-X)(z_1z_2-Z)YZ}{y_1y_2z_1(x_1y_2+z_2)X-(x_1y_2z_1+z_1z_2-Z)XY+x_2[y_1z_2+x_1(y_1y_2-Y)]Z}.\label{Tet-NLS-9D-i}
\end{align}
\end{subequations}
Map \eqref{Tet-NLS-9D} is a nine-dimensional noninvolutive tetrahedron map.
\end{proposition}
\begin{proof}
The system consisting of equations \eqref{NLS-correspondence-a} and equation $U-u_1u_2=X-x_1x_2$ has a unique solution given by \eqref{Tet-NLS-9D-a}--\eqref{Tet-NLS-9D-c}. Moreover, substituting $U$ given by \eqref{Tet-NLS-9D-c} to \eqref{NLS-correspondence-b} and \eqref{NLS-correspondence-c}, we obtain $v_1$ and $V$ given in \eqref{Tet-NLS-9D-d} and \eqref{Tet-NLS-9D-e}, respectively.
The tetrahedron property can be readily verified by substitution to the Tetrahedron equation. Finally, for the involutivity of the map we have 
$$
w_2(u_1,u_2,U,v_1,v_2,V,w_1,w_2,W)=x_1y_2+z_2+\frac{(x_2z_2+y_2X)[x_1(y_1y_2-Y)+y_1z_2]}{z_1z_2-Z}.
$$
That is, $T\circ T\neq\id$. Thus, map \eqref{Tet-NLS-9D} is noninvolutive.
\end{proof}

\subsection{Restriction on invariant leaves: A novel six-dimensional Tetrahedron map}
The existence of first integral \eqref{1stInt} indicates the integrals of map \eqref{Tet-NLS-9D}. We  can restrict the latter to a novel nine-dimensional one on the level sets of these integrals. In particular we have the following.
\begin{theorem}
\begin{enumerate}
    \item[\textbf{1.}] The quantities $\Phi=X-x_1x_2$, $\Psi=Y-y_1y_2$ and $\Omega=Z-z_1z_2$ are invariants of the map \eqref{Tet-NLS-9D}.
    \item[\textbf{2.}] Map \eqref{Tet-NLS-9D} can be restricted to a noninvolutive  parametric six-dimensional tetrahedron map 
$$
(x_1,x_2,y_1,y_2,z_1,z_2)\overset{T_{a,b,c}}{\longrightarrow}(u_1,u_2,v_1,v_2,w_1,w_2),
$$
given by:
\begin{subequations}\label{Tet-NLS} 
\begin{align}
x_1\mapsto u_1 &=\frac{b x_1-y_1z_2}{c};\\
x_2\mapsto u_2 &=\frac{ac[x_2(c+z_1z_2)+y_2z_1(a+x_1x_2)]}{abc-[x_2z_2+y_2(a+x_1x_2)][bx_1z_1-y_1(c+z_1z_2)]},\\
y_1\mapsto v_1 &=\frac{y_1(c+z_1z_2)-bx_1z_1}{ac};\\
y_2\mapsto v_2 &=x_2z_2+y_2(a+x_1x_2);\\
z_1\mapsto w_1 &=\frac{c[bz_1(a+x_1x_2)-x_2y_1(c+z_1z_2)]}{abc-[x_2z_2+y_2(a+x_1x_2)][bx_1z_1-y_1(c+z_1z_2)]};\\ 
z_2\mapsto w_2 &=x_1y_2+z_2.
\end{align}
\end{subequations}
on the invariant leaves
\begin{align}
    A_a&:=\{(x_1,x_2,X)\in\mathbb{C}^3:X=a+x_1x_2\},\quad B_b:=\{(y_1,y_2,Y)\in\mathbb{C}^3:Y=b+y_1y_2\},\nonumber\\ C_c&:=\{(z_1,z_2,Z)\in\mathbb{C}^3:Z=c+z_1z_2\}.\label{inv-leaves}
\end{align}
\item[\textbf{3.}] Map \eqref{Tet-NLS} admits the following invariants:
\begin{equation}\label{invariants}
    I_1=(a+x_1x_2)(b+y_1y_2),\qquad  I_2=(b+y_1y_2)(c+z_1z_2),\qquad I_3=ay_1y_2+(x_2y_1+z_1)(x_1y_2+z_2).
\end{equation}
\end{enumerate}
\end{theorem}
\begin{proof}
Regarding \textbf{1}. The existence of these invariants is indicated by the existence of the first integral \eqref{1stInt}. This can be verified by straightforward calculation.

With regards to \textbf{2}, we set $\Phi=a$, $\Psi=b$ and $\Omega=c$. Now, using the conditions $X=a+x_1x_2$, $Y=b+y_1y_2$ and $Z=c+z_1z_2$, we eliminate $X$, $Y$ and $Z$ from  the nine-dimensional map \eqref{Tet-NLS-9D}, and we obtain \eqref{Tet-NLS}. It can be verified by substitution that map \eqref{Tet-NLS} satisfies the parametric tetrahedron equation. For the involutivity check, we have that
$$
w_2(u_1,u_2,v_1,v_2,w_1,w_2)=x_1y_2+z_2+\frac{(bx_1-y_1z_2)[ay_2+x_2(x_1y_2+z_2)]}{c}.
$$
Thus, $T_{a,b,c}\circ T_{a,b,c}\neq \id$, and the map is noninvolutive.

Finally, regarding \textbf{3}, it can be readily proven that $(a+u_1u_2)(b+v_1v_2)\overset{\eqref{Tet-NLS}}{=}(a+x_1x_2)(b+y_1y_2)$, $(b+v_1v_2)(c+w_1w_2)\overset{\eqref{Tet-NLS}}{=}(b+y_1y_2)(c+z_1z_2)$ and $av_1v_2+(u_2v_1+w_1)(u_1v_2+w_2)\overset{\eqref{Tet-NLS}}{=}ay_1y_2+(x_2y_1+z_1)(x_1y_2+z_2)$.
\end{proof}

\begin{corollary}
Map \eqref{Tet-NLS} has the following Lax representation
\begin{equation}\label{NLS-Lax-Tet}
    M_{12}(u_1,u_2;a)M_{13}(v_1,v_2;b)M_{23}(w_1,w_2;c)= M_{23}(z_1,z_2;c)M_{13}(y_1,y_2;b)M_{12}(x_1,x_2;a),
\end{equation}
where the associated matrix is given by
\begin{equation}\label{M-NLS-Tet}
  M(x_1,x_2;a)=\left(
     \begin{array}{cc}
         a+x_1x_2 & x_1\\
         x_2 & 1
     \end{array}\right).
\end{equation}
\end{corollary}

\begin{remark}\normalfont
The invariants \eqref{invariants} are not obtained from the trace of the right-hand side of \eqref{NLS-Lax-Tet}. Matrix $M$ in \eqref{M-NLS-Tet} is not symmetric, thus does not fall in the category of maps for which proposition \ref{symm-Map-inv} holds.
\end{remark}

\begin{remark}\normalfont
As mentioned in section \ref{preliminaries} the symmetry break in equation \eqref{NLS-Lax-Tet} does not automatically imply birationality of map \eqref{Tet-NLS}. However, solving equation \eqref{NLS-Lax-Tet} for $(x_1,x_2,y_1,y_2,z_1,z_2)$, one can easily see that map \eqref{Tet-NLS} is birational.
\end{remark}


\subsection{A novel `degenerated' six-dimensional parametric tetrahedron map of NLS type}\label{Degenerated-NLS-map}
Here, we employ the second Darboux matrix of NLS type, namely matrix \eqref{M-degen}, to derive another parametric tetrahedron map. In particular, we change $(f+\lambda,p)\rightarrow (x_1,x_2)$ in \eqref{M-degen}, and define the following matrix
\begin{equation}\label{M-NLS-deg}
  M(x_1,x_2,a)=\left(
     \begin{array}{cc}
         x_1 & x_2\\
         \frac{a}{x_2} & 0
     \end{array}\right).
\end{equation}

For matrix \eqref{M-NLS-deg} we consider the following matrix trifactorisation problem:
\begin{equation}\label{NLS-Deg-Lax-Tetra}
    M_{12}(u_1,u_2,a)M_{13}(v_1,v_2,b)M_{23}(w_1,w_2,c)= M_{23}(z_1,z_2,c)M_{13}(y_1,y_2,b)M_{12}(x_1,x_2,a).
\end{equation}
The above equation implies the system of polynomial equations
\begin{equation}\label{correspondence-6D}
u_1v_1=x_1y_1,~~ u_2w_1+c\frac{u_1v_2}{w_2}=x_2y_1,~~ u_2w_2=y_2,~~ a\frac{v_1}{u_2}=a\frac{z_1}{x_2}+b\frac{x_1z_2}{y_2},~~ ac\frac{v_2}{u_2w_2}=b\frac{x_2z_2}{y_2},~~\frac{b}{v_2}=\frac{ac}{x_2z_2},
\end{equation}
which defines the following correspondence
\begin{equation}\label{corr-NLS-Deg}
u_2=\frac{a x_1 x_2y_1y_2}{u_1(ay_2z_1+bx_1x_2z_2)};~~v_1=\frac{x_1y_1}{u_1};~~ v_2 =\frac{bx_2z_2}{ac};~~w_1 =\frac{z_1}{x_1}u_1;~~ w_2 =\frac{ay_2z_1+bx_1x_2z_2}{ax_1x_2y_1}u_1.
\end{equation}
That is, $u_2,v_1,v_2,w_1$ and $w_2$ depend on $u_1$.

As in the previous section, this correspondence does not define a tetrahedron map for any choice of $u_1$. However, for the choice $u_1=y_1$ we have the following.
\begin{theorem}
The map defined by
\begin{equation}\label{Tet-NLS-Deg} 
(x_1,x_2,y_1,y_2,z_1,z_2)\overset{T_{a,b,c}}{\longrightarrow }\left(y_1,\frac{a x_1 x_2y_2}{ay_2z_1+bx_1x_2z_2},x_1,\frac{bx_2z_2}{ac},\frac{y_1z_1}{x_1},\frac{ay_2z_1+bx_1x_2z_2}{ax_1x_2}\right)
\end{equation}
is a six-dimensional parametric tetrahedron map and it is noninvolutive and birational. Moreover, map \eqref{Tet-NLS-Deg} admits the following invariants:
\begin{equation}\label{invariants-deg} 
    I_1=x_1y_1,\qquad  I_2=x_1+y_1,\qquad I_3=y_1z_1.
\end{equation}
\end{theorem}
\begin{proof}
Map \eqref{Tet-NLS-Deg} follows after substitution of $u_2=x_2$ to \eqref{corr-NLS-Deg}. The tetrahedron property can be readily verified by substitution to the tetrahedron equation. 

Now, since $v_2\circ T_{a,b,c}=\frac{by_2}{a c}\neq y_2$, it follows that $T_{a,b,c}\circ T_{a,b,c}\neq\id$, and therefore the map is noninvolutive. Additionally, the inverse of map \eqref{Tet-NLS-Deg} is given by
$$
(u_1,u_2,v_1,v_2,w_1,w_2)\overset{T_{a,b,c}^{-1}}{\longrightarrow }\left(v_1,\frac{u_2w_1w_2+cu_1v_2}{u_1w_2},u_1,u_2w_2,\frac{v_1w_1}{u_1},\frac{acu_1v_2w_2}{bcu_1v_2+bu_2w_1w_2)}\right),
$$
namely \eqref{Tet-NLS-Deg} is birational.

Finally, we have
$$
u_1v_1\overset{\eqref{Tet-NLS-Deg}}{=} x_1y_1,\quad u_1+v_1\overset{\eqref{Tet-NLS-Deg}}{=} x_1+y_1,\quad\text{and}\quad u_1w_1\overset{\eqref{Tet-NLS-Deg}}{=}x_1z_1,
$$
i.e. map \eqref{Tet-NLS-Deg} admits the invariants $I_i$, $i=1,2,3$, given by \eqref{invariants-deg}.
\end{proof}


\subsubsection{Restriction on the level sets of the invariants}
Here, we restrict the map \eqref{Tet-NLS-Deg} on the level sets of its invariants \eqref{invariants-deg}. In particular, we have the following.

\begin{proposition}
Map \eqref{Tet-NLS-Deg} can be restricted to a three-dimensional noninvolutive parametric map given by
\begin{equation}\label{Tet-NLS-Deg-Restr}
    (x,y,z)\overset{T_{a,b,c}}{\longrightarrow}\left(\frac{axy}{ay+bxz},\frac{bxz}{ac},\frac{ay+bxz}{ax}\right)
\end{equation}
\end{proposition}
\begin{proof}
Setting $I_1=1, I_2=2$ and $I_3=1$, where $I_i$, $i=1,2,3$, are given by \eqref{invariants-deg}, and solving for $x_1,y_1,z_1$ we obtain $x_1=y_1=z_1=1$. We substitute to \eqref{Tet-DNLS-Deg}, and we obtain a map $x_2\mapsto u_2(x_2,y_2,z_2)$, $y_2\mapsto v_2(x_2,y_2,z_2)$ and $z_2\mapsto w_2(x_2,y_2,z_2)$. After relabelling $(x_2,y_2,z_2,u_2,v_2,w_2)\rightarrow (x,y,z,u,v,w)$, we obtain map \eqref{Tet-NLS-Deg-Restr}. Noninvolutivity of the map follows from the fact that, for instance, $(v\circ T_{a,b,c})(x,y,z)=\frac{by}{ac}\neq y$, thus $T_{a,b,c}\circ T_{a,b,c}\neq\id$. 
\end{proof}

\begin{remark}\normalfont
From the above parametric tetrahedron map, one can obtain map (20) in Sergeev's classification \cite{Sergeev} (and also in \cite{Kashaev-Sergeev}), considering the limit $b\rightarrow a$.
\end{remark}

\subsection{A novel six-dimensional parametric tetrahedron map which does not restrict to a tetrahedron map on the level sets of its invariants}
In this section, we demonstrate by an example that a restriction of a tetrahedron map on invariant leaves is not necessarily a tetrahedron map.

In section \ref{Degenerated-NLS-map}, we saw that for the choice $u_1=y_1$, the correspondence \eqref{corr-NLS-Deg} defines the six-dimensional tetrahedron map \eqref{Tet-NLS-Deg} which can be restricted to the parametric three-dimensional tetrahedron map \eqref{Tet-NLS-Deg-Restr}. As mentioned earlier, the correspondence \eqref{correspondence-6D} does not define a parametric tetrahedron map for any choice of the free variable. And even when a six-dimensional tetrahedron map is defined, this does not necessarily imply that its three-dimensional restriction on invariant leaves has the tetrahedron property. To demonstrate this, we express $u_1, v_1, v_2, w_1$ and $w_2$ in terms of $u_2$ in \eqref{correspondence-6D}.

Specifically, the choice $u_2=x_2$ in \eqref{correspondence-6D} implies the following. 

\begin{theorem}
The map defined by
\begin{equation}\label{Tet-NLS-Deg-2} 
(x_1,x_2,y_1,y_2,z_1,z_2)\overset{T_{a,b,c}}{\longrightarrow }\left(\frac{a x_1y_1y_2}{ay_2z_1+bx_1x_2z_2},x_2,z_1+\frac{bx_1x_2z_2}{ay_2},\frac{bx_2z_2}{ac},\frac{ay_1y_2z_1}{ay_2z_1+bx_1x_2z_2},\frac{y_2}{x_2}\right)
\end{equation}
is a six-dimensional parametric tetrahedron map and it is noninvolutive and birational. Moreover, map \eqref{Tet-NLS-Deg-2} admits the following invariants:
\begin{equation}\label{invariants-deg-2}
    I_1=x_1y_1,\qquad  I_2=x_2,\qquad I_3=\frac{x_1}{z_1}.
\end{equation}
\end{theorem}
\begin{proof}
Expressing $u_1, v_1, v_2, w_1$ and $w_2$ in terms of $u_2$ in \eqref{correspondence-6D},  we obtain 
\begin{equation}\label{corr-NLS-Deg-2}
u_1=\frac{a x_1 x_2y_1y_2}{u_2(ay_2z_1+bx_1x_2z_2)};~~v_1=u_2(\frac{z_1}{x_2}+\frac{bx_1z_2}{ay_2});~~ v_2 =\frac{bx_2z_2}{ac};~~w_1 =\frac{ax_2y_1y_2z_1}{u_2(ay_2z_1+bx_1x_2z_2)};~~ w_2 =\frac{y_2}{u_2}.
\end{equation}
Map \eqref{Tet-NLS-Deg-2} follows after substitution of $u_2=x_2$ to \eqref{corr-NLS-Deg-2}. The tetrahedron property can be readily verified by substitution to the tetrahedron equation. 
Now, since $v_1(u_1,u_2,v_1,v_2,w_1,w_2)=\frac{ay_1y_2(cx_1y_2+x_2z_1z_2)}{x_2z_2(ay_2z_1+bx_1x_2z_2)}\neq y_1$, it follows that $T_{a,b,c}\circ T_{a,b,c}\neq\id$, and therefore the map is noninvolutive. Additionally, the inverse of map \eqref{Tet-NLS-Deg-2} is given by
$$
(u_1,u_2,v_1,v_2,w_1,w_2)\overset{T_{a,b,c}^{-1}}{\longrightarrow }(\frac{u_1u_2v_1w_2}{cu_1v_2+u_2w_1w_2},u_2,w_1+c\frac{u_1v_2}{u_2w_2},u_2w_2,\frac{u_2v_1w_1w_2}{cu_1v_2+u_2w_1w_2},\frac{ac v_2}{b u_2}),
$$
namely \eqref{Tet-NLS-Deg-2} is birational.

Finally, we have
$$
u_1v_1\overset{\eqref{Tet-NLS-Deg-2}}{=} x_1y_1,\quad u_2\overset{\eqref{Tet-NLS-Deg}}{=} x_2,\quad\text{and}\quad \frac{u_1}{w_1}\overset{\eqref{Tet-NLS-Deg-2}}{=}\frac{x_1}{z_1},
$$
i.e. map \eqref{Tet-NLS-Deg-2} admits the invariants $I_i$, $i=1,2,3$, given by \eqref{invariants-deg-2}.
\end{proof}

\subsection{Restriction on the level sets of the invariants}
Here, we restrict the map \eqref{Tet-NLS-Deg-2} on the level sets of its invariants \eqref{invariants-deg}. In particular, we have the following.

\begin{proposition}
Map \eqref{Tet-NLS-Deg-2} can be restricted to a three-dimensional noninvolutive parametric map given by
\begin{equation}\label{Tet-NLS-Deg-Restr-2}
    (x,y,z)\overset{T_{a,b,c}}{\longrightarrow}\left(\frac{ay+bz}{axy},\frac{bz}{ac},y\right)
\end{equation}
\end{proposition}
\begin{proof}
Setting $I_1=I_2=I_3=1$, where $I_i$, $i=1,2,3$, are given by \eqref{invariants-deg}, and solving for $x_1,x_2,z_1$ we obtain $x_1=z_1=\frac{1}{y_1}$ and $x_2=1$. We substitute to \eqref{Tet-NLS-Deg-2}, and we obtain a map $y_1\mapsto v_1(y_1,y_2,z_2)$, $y_2\mapsto v_2(y_1,y_2,z_2)$ and $z_2\mapsto w_2(y_1,y_2,z_2)$. After relabelling $(y_1,y_2,z_2,v_1,v_2,w_2)\rightarrow (x,y,z,u,v,w)$, we obtain map \eqref{Tet-NLS-Deg-Restr-2}. The noninvolutivity of the maps follows from the fact that, for instance, $(w\circ T_{a,b,c})(x,y,z)=\frac{bz}{ac}\neq z$, thus $T_{a,b,c}\circ T_{a,b,c}\neq\id$. 
\end{proof}

\begin{remark}\normalfont
Map \eqref{Tet-NLS-Deg-Restr-2} does not satisfy the parametric tetrahedron equation.
\end{remark}

\section{Derivative nonlinear Schr\"odinger type tetrahedron maps}
In this section, we study matrix trifactorisation problems \eqref{Lax-Tetra} for the Darboux transformations \eqref{DT-sl2-gen} and \eqref{DT-sl2-degen} associated to the DNLS equation.
\subsection{A novel nine-dimensional tetrahedron map}
Changing $(\lambda^{-1}p,\lambda^{-1}\tilde{q},\lambda^2 f)\rightarrow (x_1,x_2,X)$ in \eqref{DM-NLS}, we define the following matrix
\begin{equation}\label{M-DNLS}
  M(x_1,x_2,X)=\left(
     \begin{array}{cc}
         x_1 & x_1X\\
         x_2X & 1
     \end{array}\right).
\end{equation}
For matrix \eqref{M-DNLS} we consider the matrix trifactorisation problem \eqref{Lax-Tetra}, i.e.
\begin{equation}\label{DNLS-Lax-Tetra}
    M_{12}(u_1,u_2,U)M_{13}(v_1,v_2,V)M_{23}(w_1,w_2,W)= M_{23}(z_1,z_2,Z)M_{13}(y_1,y_2,Y)M_{12}(x_1,x_2,X).
\end{equation}
The above matrix equation implies the following correspondence:
\begin{subequations}\label{DNLS-correspondence}
\begin{align}
    u_1&=\frac{x_1X(y_1y_2Y-1)+y_1z_2Z}{z_1z_2Z-1}\frac{Y}{ZU},\quad u_2=(x_2+y_2z_1)\frac{Z}{Y},\label{DNLS-correspondence-a}\\
    v_1&=\frac{y_1+x_1z_1X(y_1y_2Y-1)}{1+x_2y_1z_2Z+(y_1y_2Y-1)[x_1(x_2+y_2z_1Y)X+z_1z_2 Z]}\frac{1}{X},\label{DNLS-correspondence-c}\\
    v_2&=(y_2+x_2z_2\frac{Z}{Y})U,\quad V=\frac{XY}{U},\quad w_1=\frac{x_2y_1+z_1(y_1y_2Y-1)}{x_1x_2X-1},\label{DNLS-correspondence-d}\\
    w_2&=(x_1y_2XY+z_2Z)\frac{1+x_2y_1z_2Z+(y_1y_2Y-1)[x_1(x_2+y_2z_1Y)X+z_1z_2 Z]}{(x_1x_2X-1)(z_1z_2Z-1)},\label{DNLS-correspondence-f}\\ \
    W&=\frac{(x_1x_2X-1)(z_1z_2Z-1)Z}{1+x_2y_1z_2Z+(y_1y_2Y-1)[x_1(x_2+y_2z_1Y)X+z_1z_2 Z]}.\label{DNLS-correspondence-g}
\end{align}
\end{subequations}
The elements $u_1$, $v_2$ and $V$ are expressed in terms of $U$. However, there is at least a choice of $U$ for which \eqref{NLS-correspondence} defines a tetrahedron map.

The determinant of equation \eqref{DNLS-Lax-Tetra} implies the relation
$$
(U-u_1u_2U^2)(V-v_1v_2V^2)(W-w_1w_2W^2)=(X-x_1x_2X^2)(Y-y_1y_2Y^2)(Z-z_1z_2Z^2)
$$
We choose $U-u_1u_2U^2=X-x_1x_2X^2$, $V-v_1v_2V^2=Y-y_1y_2Y^2$ and $W-w_1w_2W^2=Z-z_1z_2Z^2$. Then, the following holds.

\begin{proposition}
The system consisting of equation \eqref{DNLS-Lax-Tetra} together with $U-u_1u_2U^2=X-x_1x_2X^2$  has a unique solution given by a map 
$$
(x_1,x_2,X,y_1,y_2,Y,z_1,z_2,Z)\overset{T}{\longrightarrow}(u_1,u_2,U,v_1,v_2,V,w_1,w_2,W)
$$
defined by
\begin{subequations}\label{Tet-DNLS-9D}
\begin{align}
x_1\mapsto u_1 &=\frac{1+x_2y_1z_2Z+(y_1y_2Y-1)[x_1(x_2+y_2z_1Y)X+z_1z_2 Z]}{(x_1x_2X-1)X(z_1z_2Z-1)^2Z}\cdot[x_1X(y_1y_2Y-1)+y_1z_2Z]Y,\label{Tet-DNLS-9D-a}\\
x_2\mapsto u_2&=(x_2+y_2z_1)\frac{Z}{Y},\label{Tet-DNLS-9D-b}\\
X\mapsto U &= \frac{(x_1x_2X-1)X(z_1z_2Z-1)}{1+x_2y_1z_2Z+(y_1y_2Y-1)[x_1(x_2+y_2z_1Y)X+z_1z_2 Z]},\label{Tet-DNLS-9D-c}\\
y_1\mapsto v_1&=\frac{y_1+x_1z_1X(y_1y_2Y-1)}{1+x_2y_1z_2Z+(y_1y_2Y-1)[x_1(x_2+y_2z_1Y)X+z_1z_2 Z]}\frac{1}{X},\label{Tet-DNLS-9D-d}\\
y_2\mapsto v_2 &=\frac{(x_1x_2X-1)X(y_2Y+x_2z_2Z)(z_1z_2Z-1)}{1+x_2y_1z_2Z+(y_1y_2Y-1)[x_1(x_2+y_2z_1Y)X+z_1z_2 Z]}\frac{1}{Y},\label{Tet-DNLS-9D-e}\\
Y\mapsto V &=\frac{1+x_2y_1z_2Z+(y_1y_2Y-1)[x_1(x_2+y_2z_1Y)X+z_1z_2 Z]}{(x_1x_2X-1)(z_1z_2Z-1)Z}Y,\label{Tet-DNLS-9D-f}\\
z_1\mapsto w_1&=\frac{x_2y_1+z_1(y_1y_2Y-1)}{x_1x_2X-1},\label{Tet-DNLS-9D-g}\\
z_2\mapsto w_2&=(x_1y_2XY+z_2Z)\frac{1+x_2y_1z_2Z+(y_1y_2Y-1)[x_1(x_2+y_2z_1Y)X+z_1z_2 Z]}{(x_1x_2X-1)(z_1z_2Z-1)},\label{Tet-DNLS-9D-h}\\
z\mapsto W&=\frac{(x_1x_2X-1)(z_1z_2Z-1)}{1+x_2y_1z_2Z+(y_1y_2Y-1)[x_1(x_2+y_2z_1Y)X+z_1z_2 Z]}.\label{Tet-DNLS-9D-i}
\end{align}
\end{subequations}
Map \eqref{Tet-DNLS-9D} is a nine-dimensional birational tetrahedron map.
\end{proposition}
\begin{proof}
The system consisting of equations \eqref{DNLS-correspondence-a} and $(1-u_1u_2U)U=(1-x_1x_2X)X$ has a unique solution given by \eqref{Tet-DNLS-9D-a}, \eqref{Tet-DNLS-9D-b} and \eqref{Tet-DNLS-9D-c}. Now substituting $U$ from \eqref{Tet-DNLS-9D-c} to the first two equations of \eqref{DNLS-correspondence-d}, we obtain $v_2$ and $V$ given by \eqref{Tet-DNLS-9D-e} and \eqref{Tet-DNLS-9D-f}, respectively. 

It can be readily verified that map \eqref{Tet-DNLS-9D} is a tetrahedron map by substitution to the tetrahedron equation. Finally, we have $T\circ T\neq\id$. Thus, map \eqref{Tet-DNLS-9D} is noninvolutive.
\end{proof}

\subsection{Restriction on invariant leaves: A novel six-dimensional Tetrahedron map}

The existence of first integral \eqref{sl2-D-con-det-gen} indicates the integrals of map \eqref{Tet-DNLS-9D}. We  can restrict the latter to a novel nine-dimensional one on the level sets of these integrals. In particular, we have the following.
\begin{theorem}
\begin{enumerate}
    \item[\textbf{1.}] The quantities $\Phi=X-x_1x_2X^2$, $\Psi=Y-y_1y_2Y^2$ and $\Omega=Z-z_1z_2Z^2$ are invariants of the map \eqref{Tet-DNLS-9D}.
    \item[\textbf{2.}] Map \eqref{Tet-DNLS-9D} can be restricted to a six-dimensional parametric tetrahedron map
$$
(x_1,X,y_1,Y,z_1,Z)\overset{T_{a,b,c}}{\longrightarrow}(u_1,U,v_1,V,w_1,W),
$$
given by
\begin{subequations}\label{Tet-DNLS}
\begin{align}
x_1\mapsto u_1 &=\frac{f_{b,c}(X,Y,Z)[f_{b,a}(X,Y,X)f_{b,c}(X,Y,Z)+f_{b,0}(X,Y,-1)x_1z_1X^2YZ]}{ac^2x_1y_1z_1^2X^2Y^2Z},\\
X\mapsto U &=\frac{acx_1y_1z_1X^2Y^2}{f_{b,a}(X,Y,X)f_{b,c}(X,Y,Z)+f_{b,0}(X,Y,-1)x_1z_1X^2YZ},\\
y_1\mapsto v_1 &=\frac{x_1y_1z_1XYZf_{b,0}(X,Y,-1)}{f_{b,a}(X,Y,X)f_{b,c}(X,Y,Z)+f_{b,0}(X,Y,-1)x_1z_1X^2YZ}\\
Y\mapsto V &=\frac{f_{b,a}(X,Y,X)f_{b,c}(X,Y,Z)+f_{b,0}(X,Y,-1)x_1z_1X^2YZ}{acx_1y_1z_1XY}\\
z_1\mapsto w_1 &=\frac{f_{b,a}(X,Y,X)}{ax_1XY},\\ 
Z\mapsto W &=\frac{acx_1y_1z_1XY^2Z}{f_{b,a}(X,Y,X)f_{b,c}(X,Y,Z)+f_{b,0}(X,Y,-1)x_1z_1X^2YZ},
\end{align}
\end{subequations}
where $f_{b,c}(X,Y,Z)=bx_1z_1XZ+y_1Y(c-Z)$, on the invariant leaves
\begin{align}
    A_a&:=\{(x_1,x_2,X)\in\mathbb{C}^3:X-x_1x_2X^2=a\},\quad B_b:=\{(y_1,y_2,Y)\in\mathbb{C}^3:Y-y_1y_2Y^2=b\},\nonumber\\ C_c&:=\{(z_1,z_2,Z)\in\mathbb{C}^3:Z-z_1z_2Z^2=c\}.\label{inv-leaves-DNLS}
\end{align}
\item[\textbf{3.}] Map \eqref{Tet-DNLS} admits the following invariants:
\begin{equation}\label{invariants-DNLS}
    I_1=\frac{X}{Z},\qquad  I_2=XY,\qquad I_3=Z-bX-aYZ+\frac{y_1Y(a-X)(c-Z)}{x_1z_1XZ}-\frac{x_1z_1XZ(b-Y)}{y_1Y}.
\end{equation}
\end{enumerate}
\end{theorem}
\begin{proof}
Concerning \textbf{1}. The existence of these invariants is indicated by the existence of the first integral \eqref{sl2-D-con-det-gen}. This can be verified by straightforward calculation.

With regards to \textbf{2}, we set $\Phi=a$, $\Psi=b$ and $\Omega=c$. Now, using the conditions $x_2=\frac{X-a}{x_1X}$, $y_2=\frac{Y-b}{y_1Y}$ and $z_2=\frac{Z-c}{z_1Z}$, we eliminate $x_2$, $y_2$ and $z_2$ from  the nine-dimensional map \eqref{Tet-DNLS-9D}, and we obtain \eqref{Tet-DNLS}. It can be verified by substitution that map \eqref{Tet-DNLS} satisfies the parametric tetrahedron equation. For the involutivity check, we have that
$w_1(u_1,U,v_1,V,w_1,W)\neq z_1$, thus $T_{a,b,c}\circ T_{a,b,c}\neq \id$, and the map is noninvolutive.

Finally, regarding \textbf{3}, it can be readily proven that $\frac{U}{W}\overset{\eqref{Tet-DNLS}}{=}\frac{X}{Z}$, $UV\overset{\eqref{Tet-DNLS}}{=}XY$ and $W-bU-aVW+\frac{v_1V(a-U)(c-W)}{u_1w_1UW}-\frac{u_1w_1UW(b-V)}{v_1V}\overset{\eqref{Tet-DNLS}}{=}Z-bX-aYZ+\frac{y_1Y(a-X)(c-Z)}{x_1z_1XZ}-\frac{x_1z_1XZ(b-Y)}{y_1Y}$.
\end{proof}

\begin{corollary}
Map \eqref{Tet-DNLS} has the following Lax representation
\begin{equation}\label{DNLS-Lax-Tet}
    M_{12}(u_1,U;a)M_{13}(v_1,V;b)M_{23}(w_1,W;c)= M_{23}(z_1,Z;c)M_{13}(y_1,Y;b)M_{12}(x_1,X;a),
\end{equation}
where the associated matrix is given by
\begin{equation}\label{M-DNLS-Tet}
  M(x_1,X;a)=\left(
     \begin{array}{cc}
         X+a & x_1 X\\
         \frac{X-a}{x_1X} & 1
     \end{array}\right).
\end{equation}
\end{corollary}

\begin{remark}\normalfont
The invariants \eqref{invariants-DNLS} do not follow from the trace of the right-hand side of \eqref{DNLS-Lax-Tet}. Matrix $M$ in \eqref{M-DNLS-Tet} is not symmetric.
\end{remark}

\begin{remark}\normalfont
Although the birationality of map \eqref{Tet-DNLS} does not follow automatically from \eqref{NLS-Lax-Tet}, due to a symmetry break, solving equation \eqref{DNLS-Lax-Tet} for $(x_1,X,y_1,Y,z_1,Z)$, one can see that map \eqref{Tet-DNLS} is indeed birational.
\end{remark}

\subsection{A novel `degenerated' six-dimensional parametric tetrahedron map of DNLS type}
Now, we make use of the second Darboux matrix of DNLS type, namely matrix \eqref{DT-sl2-degen} to derive another parametric tetrahedron map. Specifically, we change $(\lambda^2f,\lambda^{-1}p)\rightarrow (x_1,x_2)$ in \eqref{DT-sl2-degen}, and define the following matrix
\begin{equation} \label{M-DNLS-deg}
M(x_1,x_2;a) := \lambda^{2}\left(\begin{array}{cc} x_1 & x_1x_2\\ \frac{a}{x_1x_2} & 0 \end{array}\right).
\end{equation}
For matrix \eqref{M-DNLS-deg} we consider the following matrix trifactorisation problem:
\begin{equation}\label{DNLS-Deg-Lax-Tetra}
    M_{12}(u_1,u_2,a)M_{13}(v_1,v_2,b)M_{23}(w_1,w_2,c)= M_{23}(z_1,z_2,c)M_{13}(y_1,y_2,b)M_{12}(x_1,x_2,a).
\end{equation}
The above equation implies the system of polynomial equations
\begin{eqnarray*}
&u_1v_1=x_1y_1,~~u_1(u_2w_1+c\frac{v_1v_2}{w_1w_2})=x_1x_2y_1,~~u_1u_2w_1w_2=y_1y_2,~~\frac{av_1}{u_1u_2}=\frac{az_1}{x_1x_2}+\frac{bx_1z_1z_2}{y_1y_2}&\\
&\frac{acv_1v_2}{u_1u_2w_1w_2}=\frac{bx_1x_2z_1z_2}{y_1y_2},~~\frac{b}{v_1v_2}=\frac{ac}{x_1x_2z_1z_2},&
\end{eqnarray*}
which defines the following correspondence
\begin{equation}\label{corr-DNLS-Deg}
u_2=\frac{ax_1^2x_2y_1^2y_2}{u_1^2z_1(ay_1y_2+bx_1^2x_2z_2)},~~v_1=\frac{x_1y_1}{u_1},~~v_2=\frac{bx_2z_1z_2}{acy_1}u_1,~~w_1=\frac{u_1z_1}{x_1},~~w_2=\frac{y_2}{x_1x_2}+\frac{bx_1z_2}{ay_1}
\end{equation}
namely, $u_2,v_1,v_2$ and $w_1$ depend on $u_1$.

As in the case of the NLS equation, this correspondence does not define a tetrahedron map for any choice of $u_1=y_1$. However, for the choice $u_1=y_1$ we have the following.

\begin{theorem}
The map defined by
\begin{equation}\label{Tet-DNLS-Deg} 
(x_1,x_2,y_1,y_2,z_1,z_2)\overset{T_{a,b,c}}{\longrightarrow }\left(y_1,\frac{ax_1^2x_2y_2}{(a y_1y_2+bx_1^2x_2z_2)z_1},x_1,\frac{bx_2z_1z_2}{ac},\frac{y_1z_1}{x_1},\frac{ay_1y_2+bx_1^2x_2z_2}{ax_1x_2y_1}\right)
\end{equation}
is a six-dimensional parametric tetrahedron map, it is noninvolutive and birational. Additionally, map \eqref{Tet-DNLS-Deg} admits the following invariants:
\begin{equation}\label{invariants-deg-DNLS} 
    I_1=x_1+y_1,\qquad  I_2=y_1z_1,\qquad I_3=\frac{x_1}{z_1}.
\end{equation}
\end{theorem}
\begin{proof}
By substitution of $u_1=y_1$ to \eqref{corr-DNLS-Deg}, we obtain map \eqref{Tet-DNLS-Deg}. The tetrahedron property can be verified by substitution to the tetrahedron equation. 

We have that $T_{a,b,c}\circ T_{a,b,c}\neq\id$, since, for instance, $v_2(u_1,u_2,v_1,v_2,w_1,w_2)=\frac{by_2}{ac}\neq y_2$. Thus, the map is noninvolutive. Moreover, the inverse of map \eqref{Tet-DNLS-Deg} is given by
$$
(u_1,u_2,v_1,v_2,w_1,w_2)\overset{T_{a,b,c}^{-1}}{\longrightarrow }\left(v_1,\frac{u_2w_1^2w_2+c v_1v_2}{v_1w_1w_2},u_1,u_2w_1w_2,\frac{v_1w_1}{u_1},\frac{acu_1v_2w_2}{bcv_1v_2+bu_2w_1^2w_2}\right),
$$
i.e. \eqref{Tet-DNLS-Deg} is birational.

Finally, we have
$$
u_1+v_1\overset{\eqref{Tet-DNLS-Deg}}{=} x_1+y_1,\quad v_1w_1\overset{\eqref{Tet-DNLS-Deg}}{=} y_1z_1,\quad\text{and}\quad \frac{u_1}{w_1}\overset{\eqref{Tet-DNLS-Deg}}{=}\frac{x_1}{z_1},
$$
i.e. map \eqref{Tet-NLS-Deg} admits the invariants $I_i$, $i=1,2,3$, given by \eqref{invariants-deg-DNLS}.
\end{proof}


\subsection{Restriction on the level sets}
Now, we restrict the map \eqref{Tet-DNLS-Deg} on the level sets of its invariants \eqref{invariants-deg-DNLS}. Specifically, we have the following.

\begin{proposition}
Map \eqref{Tet-DNLS-Deg} can be restricted to a three-dimensional noninvolutive parametric tetrahedron map given by \eqref{Tet-NLS-Deg-Restr}.
\end{proposition}
\begin{proof}
Setting $I_1=1$, $I_2=\frac{1}{2}$ and $I_3=\frac{1}{2}$, where $I_i$, $i=1,2,3$, are given by \eqref{invariants-deg-DNLS}, we obtain $x_1=y_1=\frac{1}{2}$ and $z_1=1$. We substitute to \eqref{Tet-DNLS-Deg}, and we obtain a map $x_2\mapsto u_2(x_2,y_2,z_2)$, $y_2\mapsto v_2(x_2,y_2,z_2)$ and $z_2\mapsto w_2(x_2,y_2,z_2)$. After relabelling $(x_2,y_2,z_2,u_2,v_2,w_2)\rightarrow (x,y/2,z,u/2,v,w)$, we obtain map \eqref{Tet-NLS-Deg-Restr}.
\end{proof}

\section{Conclusions}
In this paper, we present novel solutions to the functional tetrahedron and parametric tetrahedron equation which are noninvolutive and birational. Noninvolutivity is an important property for a map, since involutive maps have trivial dynamics. 

Our approach is based on the study of matrix trifactorisation problems for Darboux transformations, and we study the cases of Darboux transformations associated with the NLS and DNLS equation. Specifically, in the NLS case, we construct the tetrahedron map \eqref{Tet-NLS-9D} which can be restricted to a parametric tetraedron map \eqref{Tet-NLS}. Additionally, we derive map \eqref{Tet-NLS-Deg} which can be restricted to map \eqref{Tet-NLS-Deg-Restr} on the level sets of its invariants. The latter at a certain limit gives map (20) in Sergeev's classification \cite{Sergeev}.  In DNLS case, we construct the tetrahedron map \eqref{Tet-DNLS-9D} which can be restricted to a parametric tetraedron map \eqref{Tet-DNLS}. Moreover, we construct map \eqref{Tet-DNLS-Deg} which can be also restricted on the level sets of its invariants to map \eqref{Tet-NLS-Deg-Restr}.


Our results could be extended in the following way:
\begin{enumerate}
   \item Study the integrability of the derived maps;
    \item Study the more general matrix trifactorisation problem \eqref{Lax-Tetra} with a spectral parameter where matrices $L_{ij}$, $i,j=1,2,3$, $i<j$, depend on a spectral parameter $\lambda$;
    \item Find solutions to the entwining parametric tetrahedron equation;
    \item Study the corresponding $3D$-lattice equations;
    \item Extend the results to the case of Grassmann algebras.
\end{enumerate}

Regarding 1, the existence of invariants \eqref{invariants}, \eqref{invariants-deg}, \eqref{invariants-DNLS} and \eqref{invariants-deg-DNLS} for maps \eqref{Tet-NLS}, \eqref{Tet-NLS-Deg},  \eqref{Tet-DNLS} and  \eqref{Tet-DNLS-Deg}, respectively, is already a sign of integrability of these maps. In fact, in all these cases, the derived invariants are functionally independent and they are in number half as the dimension of the corresponding maps. In order to claim complete integrability in the Liouville sense, one needs to find a Poisson bracket with respect to which the invariants are in involution.

Concerning 2, although the Darboux transformations employed in this paper depend on a spectral parameter $\lambda$, in the consideration of the associated matrix refactorisation problem \eqref{Lax-Tetra} the spectral parameter was rescaled; see matrices \eqref{M-NLS},  \eqref{M-NLS-deg},  \eqref{M-DNLS} and  \eqref{M-DNLS-deg}. In fact, if we let the latter matrices depend explicitly on the spectral parameter, then equation \eqref{Lax-Tetra} has no solutions for $u$ and $v$. The question arises as to whether equation \eqref{Lax-Tetra} defines tetrahedron maps for certain matrices $A, B, C$ and $D$ which explicitly depend on a spectral parameter $\lambda$.

With regards to 3, we mean the following equation
$$
Q^{123}\circ R^{145} \circ S^{246}\circ T^{356}=T^{356}\circ S^{246}\circ R^{145}\circ Q^{123},
$$
and its parametric version
$$
Q^{123}_{a,b,c}\circ R^{145}_{a,d,e} \circ S^{246}_{b,d,f}\circ T^{356}_{c,e,f}=T^{356}_{c,e,f}\circ S^{246}_{b,d,f}\circ R^{145}_{a,d,e}\circ Q^{123}_{a,b,c}.
$$
For $Q\equiv R\equiv S \equiv T$, we obtain equations \eqref{Tetrahedron-eq} and \eqref{Par-Tetrahedron-eq}, respectively.  Similarly, to \cite{Sokor-Pap}, we believe that solutions to these equations can be obtained by studying trifactorisation problems of Darboux transformations together with their degenerated versions. For $Q\equiv R\equiv S$, a Hirota type entwining tetrahedron map appears in \cite{Kassotakis}.

Regarding 4, it makes sense to study whether the maps derived in this paper are related to certain integrable $3D$-lattice equations. Some of the available methods in the literature are the following: i. The relation between tetrahedron maps and $3D$-lattice equations via symmetries of the latter which was established in \cite{Kassotakis-Tetrahedron}. ii. In \cite{Kouloukas-Dihn, Pap-Tongas} it was demonstrated how to lift lattice equations to Yang--Baxter maps and these ideas can be employed to the case of tetrahedron maps. iii. Since the invariants of the maps \eqref{Tet-NLS-Deg}, \eqref{Tet-NLS-Deg-2}, \eqref{Tet-NLS-Deg}, \eqref{Tet-DNLS-Deg} are in separated form, one may employ the method presented in \cite{Pavlos-Maciej, Pavlos-Maciej-2}.

Concerning 5, the extension of integrable lattice equations to the case of Grassmann algebras via Grassmann extended Darboux transformations in \cite{Georgi-Sasha} motivated the initiation of the extension of the theory of Yang--Baxter maps to the noncommutative (Grassmann) case \cite{GKM, Sokor-Kouloukas, Sokor-Sasha-2, Sokor-2020}. The study of matrix trifactorisation problems \eqref{Lij-mat} for Grassmann extended Darboux transformations may lead to solutions of the Grassmann extended tetrahedron and parametric tetrahedron equation. Fully noncommutative versions of tetrahedron maps are found in \cite{Doliwa-Kashaev}.

\section{Acknowledgements}
This work was funded by the Russian Science Foundation (project number 20-71-10110). I would like to thank P. Kassotakis and A.V. Mikhailov for useful discussions, and D. Talalaev for useful discussions and introducing me to the tetrahedron equation.

\end{document}